\documentclass[11pt]{article}

\usepackage[
	persons={Moses,Paul,Tianyu,June},
	bibliosources={refs.bib}
]{pomegranate}
\usepackage{appendix}

\DeclareOperator{\PP}
\DeclareOperator{\TV}
\DeclareOperator{\mix}
\DeclareOperator{\Cov}
\DeclareOperator{\polylog}
\DeclareOperator{\KL}
\DeclareOperator{\diag}
\DeclareOperator{\sign}
\DeclareOperator{\Var}
\DeclareOperator{\Eng}
\DeclareDocumentMathCommand{\DKL}{}{\D_\KL}
\DeclareDocumentMathCommand{\corMat}{}{\Psi^{\text{cor}}}

\graphicspath{{figures/}}
\usepackage{subcaption}

\title{On the Complexity of Sampling Redistricting Plans}

\author[1]{Moses Charikar}
\author[1]{Paul Liu}
\author[1]{Tianyu Liu}
\author[1]{Thuy-Duong Vuong}

\affil[1]{Stanford University, \url{{moses,tliu}@cs.stanford.edu,{paul.liu,tdvuong}@stanford.edu}}

\begin{document}

\maketitle
\begin{abstract}
A crucial task in the political  redistricting problem is to sample redistricting plans i.e. a partitioning of the graph of census blocks into districts.
    We show that ReCom [DeFord-Duchin-Solomon'21]-a popular Markov chain to sample redistricting plans-is exponentially slow mixing on a simple subgraph of $\Z_2.$ We show an alternative way to sample balanced, compact and contiguous redistricting plans using a "relaxed" version of ReCom and rejection sampling.
\end{abstract}
\section{Introduction}
Redistricting is the task of redrawing district boundaries, i.e. partitioning a set of geographical units into a fixed number of parts (districts), subject to certain constraints on the balance, contiguity and compactness of the partitioning. 
In the United States for example, each state is divided into congressional districts, each of which elects a representative to the US House of Representatives.
The census every 10 years triggers a flurry of redistricting activity as the district boundaries and the number of districts themselves change in response to shifts in the population.
Changing district boundaries comes with consequences for political parties, candidates and voter rights. 
The power to redraw district boundaries has been abused to favor one party over another, increase/decrease the influence of groups of people and so on -- a practice known as gerrymandering.
Understanding, detecting and reasoning about gerrymandering are issues at the intersection of mathematics, geography, political science and law --
this recent book \cite{duchin2021political} gives an accessible introduction to these many dimensions.


Detecting gerrymandering in a particular redistricting proposal is a complex question and several metrics have been proposed in the literature to measure compactness of districts and fairness of plans.
In the past few years, a methodology based on MCMC sampling has emerged as a basis for evaluating redistricting plans and detecting gerrymandering \cite{DeFord2021Recombination,autry2021metropolized,herschlag2020quantifying,fifield2020automated,mccartan2020sequential}.
The idea is to draw an ensemble of plans according to a certain distribution then compare a given redistricting plan of interest against the ensemble of these sampled plans. 
The ensemble of redistricting plans gives an estimate of reasonable ranges of values for various metrics of interest and is useful in detecting outliers.
Legal arguments based on this methodology have been made in several court cases in recent years.

Many currently used redistricting algorithms work with a discrete formulation of the redistricting problem: An instance of the problem is a graph -- a dual graph of the partition into geographic units -- with vertices corresponding to units, and edges corresponding to pairs of adjacent units.
Vertices are associated with the population of the corresponding geographic unit. 
The objects of interest are partitions of this graph into connected pieces, with additional constraints: the pieces should have roughly equal size, should be compact (have relatively small boundaries), and so on.

An important issue here is the choice of distribution over partitions to sample from.
The naive approach of sampling from the uniform distribution of connected partitions yields unreasonable redistricting plans -- most of the probability mass is concentrated on partitions that are far from compact.
Moreover, there are complexity lower bounds that rule out efficient sampling from the uniform distribution \cite{cohen2020computational}.
Previous work shows that drawing plans from the \emph{spanning tree distribution} (\cref{def:spanning tree dist}) ensures desirable properties such as contiguity and compactness \cite{DeFord2021Recombination}. Intuitively speaking, the spanning tree distribution favors partitioning of the graph where the induced subgraph on each part of the partition containing many spanning trees, which ensures compactness, i.e. there are few edges across different parts in the partition  \cite[see][]{procaccia2022compact}. Hence, we are interested in sampling from the \emph{balanced spanning tree distribution} i.e. the distribution induced by the spanning tree distribution on balanced partitions i.e. partitions where each part has exactly the same number of vertices.

\begin{definition}[Spanning tree distribution] \label{def:spanning tree dist}
Given a graph $ G= G(V,E)$ and parameter $k$, the {\bf spanning tree distribution} $\mu^*$ over the partitioning of the vertices of $G$ has density proportional to the product of number of spanning trees in each part of the partition
\[ \mu^*_{G,k}((P_1,\cdots, P_k)) = \prod_{i=1}^k T(P_i)\]
where $P_1 \cup \cdots, P_k = V$ is a partition of $V$ and $T(P_i)$ is number of spanning tree in $G[P_i].$

The {\bf balanced spanning tree distribution} $\mu^{\text{balanced}}_{G,k}$ is the distribution over balanced partitions (i.e. partitions $(P_1,\cdots, P_k)$ with $|P_1| = \cdots = |P_k|$) induced by $\mu^*_{G,k}.$
\end{definition}

A widely-used method to draw or sample redistricting plans is to run a Markov chain called ReCom (see \cref{def:ReCom} for the full definition). 
As opposed to the natural family of Flip chains, which make changes to single nodes at a time, ReCom makes more substantial moves in the space of partitions in each step.
ReCom iteratively merges two adjacent parts of the partition and splits them into two new parts while maintaining that each part of the partition has the same number of vertices. ReCom exhibits very good empirical performance, but so far there has been no rigorous theoretical study on the mixing time of ReCom, even on very simple planar graphs. 
While we do not have a closed form expression for the stationary distribution of ReCom, it targets the balanced spanning tree distribution $\mu^{\text{balanced}}_{G,k}$ defined above.
A reversible variant of ReCom was introduced by \cite{ReversibleReCom} who showed that the stationary distribution is the spanning tree distribution $\mu^*_{G,k}$.
Here are some open problems about the ReCom chain and the spanning tree distribution it targets:
\begin{question}
\label{qn:ReComMixing}
Is there a polynomial upper bound on the mixing time for the ReCom chain?
\end{question}
\begin{question}
Can we sample efficiently from the spanning tree distribution?
\label{qn:SpanningTreeSampling}
\end{question}
\begin{question}
Can we sample efficiently from the balanced spanning tree distribution?
\label{qn:BalancedSpanningTreeSampling}
\end{question}

In this work, we answer \cref{qn:ReComMixing} in the negative and show that ReCom is very slow mixing: its mixing time is exponential in the number of vertices in the graph, even for simple planar graphs that are subgraphs of the 2-dimensional grid graph $\Z^2$ (see \cref{sec:slow mixing}).

Below we define the Markov chain ReCom to sample redistricting plans. For technical reasons, we maintain a spanning tree $T_i$ for each part $P_i$ of the partition. The Markov chain will mutate the tuple $(T_1, \cdots, T_j)$ according to some probabilistic rules for some number of steps, and output $(V(T_1), \cdots, V(T_k))$ as the graph partition.
\begin{definition}[ReCom: A Markov chain to sample redistricting plans] \label{def:ReCom}
Each step of ReCom modifies the tuple $(T_1, \cdots, T_k)$ of spanning trees as followed:
\begin{enumerate}
     \item Uniformly at random, add one edge $e = (u,v)$ that connects to spanning tree $T_i$ and $T_j$ i.e. $u\in V(T_i)$ and $v\in V(T_j)$ or vice versa. Let $T:=T_i \cup T_j \cup \set*{e}$ be the new spanning tree on $ G[V (T_i )\cup V(T_j)]$ formed by adding this edge.
    \item Let $F$ be the set of edges $f\in T$ such that removing $f$ from the tree  splits the tree $T$ into two new tree $T'_i$ and $T'_j$ containing the same number of vertices. While $F \setminus \set*{e} = \emptyset$, resample a random spanning tree on $V(T_i) \cup V(T_j)$ and set $T$ to be this new tree. 
    \item Choose a random edge $f$ among $F$, remove it from $T,$ and replace $T_i$ and $T_j$ with the two newly created trees $T'_i$ and $T'_j.$
\end{enumerate}
\end{definition}

Our slow mixing results are proved using standard conductance arguments~\cite{LP17}, namely to partition the state space of the Markov chain into at least three subsets, with one of them having exponentially small weight compared with others, and show that removing this low-weight subset disconnects the state space. As a consequence, if the Markov chain enters a subset with large weight, it can hardly go through the low-weight bottleneck to reach states in other large-weight subsets, hence mixing slowly. In the case of ReCom, the partitions in the low-weight subset will consist of parts that contain few spanning trees, e.g. a ``line segment''. We show that since ReCom only moves among balanced partitions, such low-weight subset form a geometric barrier for the high-weight subsets to be connected.

We remark that originally ReCom was defined such that in each step a spanning forest of the partition is not maintained~\cite{DeFord2021Recombination}. The version of ReCom described in \cref{def:ReCom} has been proposed in subsequent works (e.g. in the name of \emph{Forest ReCom} \cite{autry2021metropolized}) to reduce computational cost.
However, our slow mixing results in \cref{sec:slow mixing} hold, regardless of whether a spanning forest is maintained or not.

On the other hand, we answer \cref{qn:SpanningTreeSampling} positively: we show that a variant of ReCom (\cref{def:variant ReCom})  converges to the spanning tree distribution in quasi-linear steps in the number of edges in the graph. Hence, we can sample from the spanning tree distribution in nearly linear time.
Briefly speaking, this ReCom variant is simply the \emph{up-down walk} \cite{ALOVV21} on the distribution of $(n-k)$-edges forest, which is known to mix in $O(m\log m)$-step where $m$ is the number of edges in the graph and each step can be implemented in amortized $O(\log m)$ time.
In addition, we introduce variations of the spanning tree distribution that favor balanced partitions (see \cref{def:c biased spanning tree distribution}), and show efficient Markov chain(s) to sample from these distributions via Markov chain comparison techniques~\cite{DS94}. See \cref{sec:relax ReCom} for details.
\begin{definition}[$c$-biased spanning tree distribution] \label{def:c biased spanning tree distribution}
Given a graph $G= G(V,E)$ and number of partitions $k$, consider the following variant of the spanning tree distribution defined by,
\[\mu^c_k((P_1,\cdots, P_k)) = \prod_{i=1}^{k} T(P_i) \abs{P_i}^c\]
where $c \geq 0$ is a given parameter, $\abs{P_i}$ denote the number of vertices in $P_i.$
\end{definition}
Observe that for $c = 0$, the $c$-biased spanning tree distribution is precisely the spanning tree distribution. Since the product $\prod_{i=1}^k \abs{P_i}$ becomes larger as the partition becomes more balanced, the $c$-biased distribution puts the most weight on partitions where each part has exactly the same size. As $c$ approaches infinity, the $c$-biased spanning tree distribution becomes close to the balanced spanning tree distribution. Hence, the $c$-biased spanning tree distribution enforces a "soft" constraint on the balancedness of the partition.

We implement our sampling algorithms on the square grid graph and on real world graphs. We empirically verify that our intuition about $c$-biased distribution is correct: for $c$ large enough, drawing from the $c$-biased distribution is effectively the same as drawing from the balanced spanning tree distribution (see \cref{fig:partition}). We also consider an alternative approach to sample from the balanced spanning tree distribution by rejection sampling i.e. drawing multiple samples from the spanning tree distribution and accepting only the balanced partitions. This strategy is efficient when the proportion of balanced partition under the spanning tree distribution is large (i.e. at least inversely polynomial in the number of vertices). We prove that this is indeed the case (see \cref{thm:percentage of balance partition simple graph}) for our hard instance where ReCom takes exponential time to mix, showing that we can efficiently sample from the balanced spanning tree distribution in that case. Empirical evidence suggests that the proportion of balanced partitions is large when the given graph is a subgraph of the grid graph and the number of partitions is a constant\footnote{This is the case of interest for some redistricting problems. There are inherent barriers against efficiently sampling redistricting plans when $k$ is super-constant \cite[see][]{cohen2020computational}}. 
\begin{figure}[h]
\centering
           \subfloat[c=2]{%
              \includegraphics[width =0.49\textwidth]{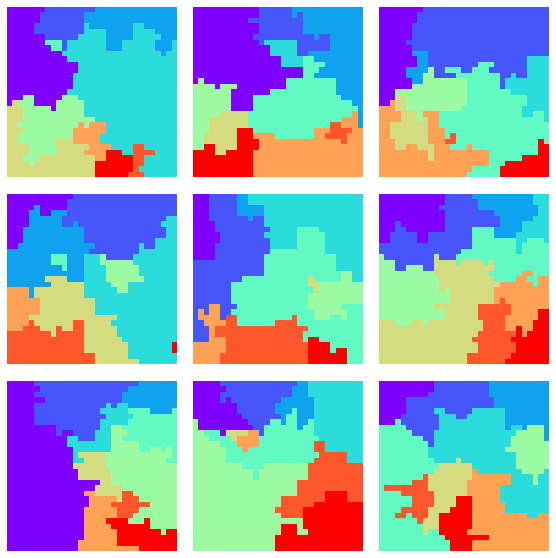}%
           } 
           \hfill
           \subfloat[c=20]{%
              \includegraphics[width=0.49\textwidth]{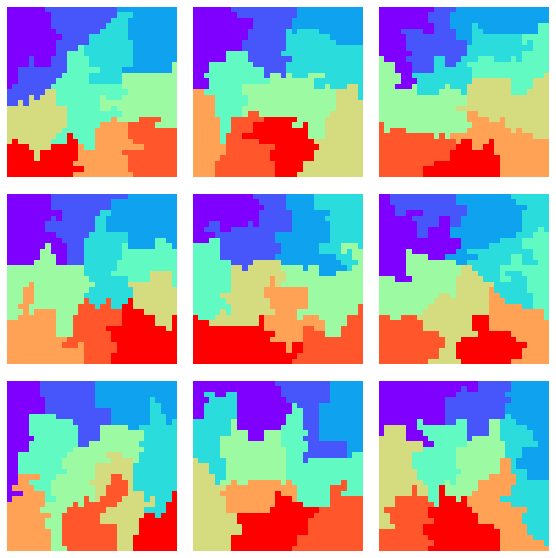}
              }
              \caption{Samples from $\mu^c$ on $30 \times 30$ grid for different parameter $c$}
           \label{fig:partition}
\end{figure}
\begin{conjecture} \label{conj:percentage of balance partition}
For the $m\times n$ grid graph, the proportion of balanced partitions under the spanning tree distribution $\mu^*_k$ is at least $1/poly(m, n)$ when $ k = O(1).$
\end{conjecture}
We prove the conjecture for rectangular grid graphs with one side being a constant (see \cref{thm:percentage of balance partition simple graph}). Proving the conjecture for general rectangular grid graphs, or even the $n\times n$ grid graph, is an interesting open problem.
\paragraph{Concurrent work}
Very recently, \cite{FW22} prove upper and lower bound for the mixing time of another Markov chain to \emph{uniformly} sample connected partitions of subgraph of the grid graph $\Z^2.$ Their work differs from ours in several aspects:
\begin{itemize}
    \item They consider the problem of sampling from the uniform distribution over partitions of a graph into connected roughly equal-sized parts. The problem of sampling partitions into exactly equal-sized parts is known to be NP-hard even for planar graphs when $k=2$ \cite{najt2019complexity}. On the other hand, so far there has been no hardness result for sampling from the balanced spanning tree distribution on planar graphs\footnote{We note that the decision problem of whether a graph can be partitioned into exactly equal-sized parts is NP-hard for $k\ge2$~\cite{dyer1985complexity}.}. Furthermore, the uniform distribution over connected partitions does not favor compact partitions. 
    \item Their fast mixing result only applies when the number of partitions is large, which is usually not the interesting case for redistricting.
    \item They prove that a chain similar to the Flip walk\footnote{Each step of the Flip walk reassigns a single vertex from one part of the partition to another part while maintaining connectedness.} is exponentially slow to mix on $2$-partitions of a subgraph of $\Z^2.$ Our slow mixing result is for ReCom, which can make global moves, and is thus potentially much faster than the Flip walk. The Flip walk has been previously shown to be slow mixing on graph families of interest, prompting the study of ReCom \cite[see][for references]{DeFord2021Recombination}.
\end{itemize}
\paragraph{Organization}
We start with some mathematical preliminaries in \cref{sec:prelims}.
In \cref{sec:slow mixing}, we present simple planar graphs where ReCom has mixing time exponential in the number of vertices.
In \cref{sec:relax ReCom} we show rapid mixing for a variant of ReCom that samples from the spanning tree distribution without balance constraint. 
Finally, in \cref{sec:FractionBalanced}, we discuss the fraction of balanced partitions in the spanning tree distribution.

\paragraph{Acknowledgement}
We thank Wesley Pegden and Weiming Feng for pointing out a mistake in the definition of the ReCom variant, and Gabe Schoenbach for pointing out typos in the previous version of this manuscript.

\section{Preliminaries}\label{sec:prelims}
For graph $G = G(V,E)$ and $U \subseteq V$ let $G[U]$ denote the subgraph induced by $G$ on $U.$

For density function $\mu:\Omega \to \R_{\geq 0},$ let $Z_{\mu}$ be its partition function i.e. \[Z_{\mu} = \sum_{\omega \in \Omega} \mu(\omega) .\]

\subsection{Markov Chains and Mixing Time}
\begin{definition}
Let $\mu, \nu$ be two discrete probability distributions over the same event space $\Omega$. The \emph{total variation distance}, or \emph{$\TV$-distance}, between $\mu$ and $\nu$ is given by
\[
\norm{\mu - \nu}_{\TV} = \frac{1}{2} \sum_{\omega \in \Omega} \abs{\mu(\omega) - \nu(\omega)}
\]
\end{definition}
\begin{definition}
Let $P$ be an ergodic Markov chain on a finite state space $\Omega$ and let $\mu$ denote its (unique) stationary distribution. For any probability distribution $\nu$ on $\Omega$ and $\epsilon \in (0,1)$, we define
\[t_\mix(P, \nu, \epsilon) = \min\set{t\geq 0 \mid \norm{\nu P^{t}- \mu}_\TV \leq \epsilon},\]
and
\[t_\mix(P,\epsilon) = \max\set*{t_\mix(P, \1_x, \epsilon) \given x\in \Omega},\]
where $\1_{x}$ is the point mass distribution supported on $x$. 
\end{definition}

We will drop $P$ and $\nu$ if they are clear from context. Moreover, if we do not specify $\epsilon$, then it is set to $1/4$. This is because the growth of $t_{\operatorname{mix}}(P,\epsilon)$ is at most logarithmic in $1/\epsilon$ (cf.~\cite{LP17}). 

The modified log-Sobolev constant of a Markov chain, defined next, provides control on its mixing time. For a detailed coverage see \cite{LP17}.
\begin{definition}
Let $P$ denote the transition matrix of an ergodic, reversible Markov chain on $\Omega$ with stationary distribution $\mu$.

\begin{itemize}
    \item The \emph{Dirichlet form} of $P$ is defined for $f,g \in \Omega \to \mathbb{R}$ by
    \[\mathcal{E}_P(f,g) = \langle f, (I-P)g \rangle_{\mu} = \langle (I-P)f, g \rangle_{\mu}.\]
    \item The \emph{modified log-Sobolev (mLSI) constant} of $P$ is defined to be
\[\rho_0(P) = \inf\set*{\frac{\mathcal{E}_P(f, \log f)}{2\cdot \text{Ent}_{\mu}[f]} \given f \colon \Omega \to \mathbb{R}_{\geq 0}, \text{Ent}_{\mu}[f]\neq 0},\]
where 
\[\text{Ent}_{\mu}[f] = \E_{\mu}{f\log f} - \E_{\mu}{f}\log\E_{\mu}{f}.\]
Note that, by rescaling, the infimum may restrict our attention to functions $f\colon \Omega \to \mathbb{R}_{\geq 0}$ satisfying $\text{Ent}_{\mu}[f]\neq 0$ and $\E_{\mu}{f} = 1$.
\item The \emph{Poincare constant} of P is defined to be
\[\alpha_0(P) = \inf\set*{\frac{\mathcal{E}_P(f,  f)}{ \text{Var}_{\mu}[f]} \given f \colon \Omega \to \mathbb{R} }\]
where
\[\text{Var}_{\mu}[f] = \E_{\mu}{f^2} - (\E_{\mu}{f})^2 = \sum_{x, y} \pi(x) \pi(y) (f(x)-f(y))^2. \]
When $P$ is reversible and the state space $\Omega$ is finite, $\alpha_0(P)$ is precisely the second eigenvalue gap i.e. $ 1-\lambda_2(P)$ where $1= \lambda_1 \geq \lambda_2 \geq \cdots \geq \lambda_{\abs{\Omega}}$ are the eigenvalues of $P.$
\end{itemize}
\end{definition}
The relationship between the modified log-Sobolev constant, the Poincare constant and mixing times is captured by the following well-known lemma. 

\begin{lemma}[cf.~\cite{bobkov2006modified}] \label{thm:mixing}
Let $P$ denote the transition matrix of an ergodic, reversible Markov chain on $\Omega$ with stationary distribution $\mu$ and let $\rho_0(P), \alpha_0(P)$ denote its modified log-Sobolev constant and Poincare constant resp. Then $ \rho_0(P) \leq \alpha_0(P)$,
\[t_{\operatorname{mix}}(P,\epsilon) \leq  \ceil*{ \alpha_0(P)^{-1}\cdot \parens*{\log\parens*{\frac{1}{\min_{x\in \Omega}\mu(x)}} + \log\parens*{\frac{1}{\epsilon}}} } \]
and
\[t_{\operatorname{mix}}(P,\epsilon) \leq  \ceil*{ \rho_0(P)^{-1}\cdot \parens*{\log\log\parens*{\frac{1}{\min_{x\in \Omega}\mu(x)}} + \log\parens*{\frac{1}{2\epsilon^2}}} }\]
\end{lemma}


The conductance\footnote{also known as bottleneck ratio in \cite{LP17}} of a subset $S$ of states in a Markov chain is
\[\Phi(S) = \frac{Q(S, \Omega \setminus S)}{\mu(S)}\]
where  $Q(S, \Omega\setminus S) = \sum_{x\in S, y \in \Omega\setminus S} \mu(x) P(x,y)$ is the ergodic flow between $S$ and $\Omega\setminus S,$ and $\mu(S) = \sum_{x\in S} \mu(x).$ The conductance of a Markov chain is defined as the minimum conductance
over all subsets $S$ with $\mu(S) \leq 1/2,$ i.e.,
\[\Phi = \min_{S: \mu(S) \leq 1/2} \Phi(S).\]


 We can lower bound the mixing time by conductance as follow.
\begin{theorem}[{\cite[see, e.g.,][Thm. 7.4]{LP17}}]
For a reversible Markov chain $P$ with conductance $\Phi,$ we have 
\[ t_{\mix}(P, 1/4) \geq \frac{1}{4 \Phi}.\]
\end{theorem}





\section{ReCom is torpidly mixing} \label{sec:slow mixing}
Observe that in the simple case of sampling $k$-partitions on a cycle graph (assume they exist), the ReCom chain is frozen no matter where it starts from. That is to say, on such ``single-cycle'' graphs, the transition graph of ReCom is no longer strongly connected. This property of Markov chain is often referred to as being \emph{reducible}, and is undesirable for the purpose of sampling.
Even if a Markov chain is irreducible on a graph, in practice it is not efficient to sample from the stationary distribution (approximately) if the mixing time is exponential in the size of the input.

In this section, we describe families of natural planar graphs with bounded degree, on which ReCom for $3$-partition takes at least exponential time to mix, although being irreducible. In particular, we show that ReCom is torpidly mixing on a family of subgraphs of the grid graph $\Z^2.$

\begin{definition}[Double-cycle graphs]
A graph $G = (V, E)$ is called a \emph{double-cycle graph} (of length $n$) if $|V| = 2n$ for some $n > 0$, and we can write $V = L \cup R$ where $L = \{l_0, l_1, \dots, l_{n-1}\}$ and $R = \{r_0, r_1, \dots, r_{n-1}\}$ such that $E = E_L \cup E_= \cup E_R$ where $E_L = \{\{l_i, l_{i+1 \pmod n}\} \mid 0 \le i \le n - 1\}$, $E_R \{\{r_i, r_{i+1 \pmod n}\} \mid 0 \le i \le n - 1\}$, and $E_= = \{\{l_i, r_i\} \mid 0 \le i \le n-1\}$.
\end{definition}
\begin{definition}[Grid-with-a-hole graphs]
The \emph{$m \times n$ grid graph} $G_{m, n}$ can be defined as $G_{m, n} = (V, E)$ where $V = \{(i, j) \mid i \in [m], j \in [n]\}$ and $E = \{\{(i, j), (i+1, j)\} \mid i \in [m-1], j \in [n]\} \cup \{\{(i, j), (i, j+1)\} \mid i \in [m], j \in [n-1]\}$.
The \emph{$m \times n$ grid-with-a-hole graph} $G^\square_{m, n}$ ($m, n \ge 4$) can be defined as $G^\square_{m, n} = G_{m, n} \setminus \{(i, j) \mid 3 \le i \le m-2, 3 \le j \le n-2\}$.
\end{definition}

\begin{theorem}\label{thm:double-cycle}
ReCom for 3-partitions mixes torpidly on the double-cycle graphs.
\end{theorem}

At a high level, to prove the torpid mixing of a Markov chain we can use the following strategy:
(1) partition the state space $\Omega$ into three disjoint subsets $\Omega_{\text{LEFT}} \cup \Omega_{\text{MID}} \cup \Omega_{\text{RIGHT}}$;
(2) show that in order to go from states in $\Omega_{\text{LEFT}}$ to states in $\Omega_{\text{RIGHT}}$, the Markov chain has to go through the ``middle states'' $\Omega_{\text{MID}}$; and
(3) demonstrate that $\mu(\Omega_{\text{MID}})$ is exponentially small (compared with $\min(\mu(\Omega_{\text{LEFT}}), \mu(\Omega_{\text{RIGHT}}))$) in the input size.
This means that starting from any state in $\Omega_{\text{LEFT}}$, the probability of going through $\Omega_{\text{MID}}$ (and consequently to any state in $\Omega_{\text{RIGHT}}$ and reach stationarity) is exponentially small. Hence the conclusion of torpid mixing.

Following this strategy, in the proof of Theorem \ref{thm:double-cycle} we show how to decompose the state space of 3-partitions for any double-cycle graph $G = (V=L \cup R, E = E_L \cup E_= \cup E_R)$. Here we assume $|L| = |R| = 3n$ for some $n$, which is necessary for the existence of 3-partitions. This implies $|V| = 6n$, and the size of each component is $2n$.
The intuition is as follows.

For a 3-partition $X = (X_1, X_2, X_3)$, let $ \phi(X) = \sum_{i=1}^3 |E_= \cap E(X_i)|$.
The stationary distribution $\mu$ (favoring partitions with more spanning trees) on a partition $X$ is roughly proportional to $r^{\phi(X)}$ for some constant $r > 2$.
In particular, 3-partitions that possess the highest weights are exactly those containing all the edges in $E_=$. Note that there are $n$ such states.
The key observation is that in order to go from any highest-weight state to another, ReCom has to go through a state $X = (X_1, X_2, X_3)$ with at least one of the $X_i$'s being a ``line segment'' completely in $L$ or completely in $R$.
As a consequence, any such state $X$ has $\phi(X) \le n$ and thus has relatively small weights (exponentially small compared with the highest weight). More importantly, the number of these states can be upper bounded by some polynomial in $n$.
Next we formalize this idea.

\begin{proof}[Proof of Theorem \ref{thm:double-cycle}]
For any $0 \le j \le n-1$, let $E_L^j = E_L \setminus \{\{l_{i*n + j - 1 \pmod {3n}}, l_{i*n + j}\} \mid i = 0, 1, 2\}$ and similarly $E_R^j = E_R \setminus \{\{r_{i*n + j - 1 \pmod {3n}}, r_{i*n + j}\} \mid i = 0, 1, 2\}$.
For $0 \le j \le n-1$, denote by $A_j$ the 3-partition $A_j = \left(V, E_L^j \cup E_= \cup E_R^j \right)$.
For edges in $E_L$ and $E_R$, label the edge $\{l_{j-1 \pmod {3n}}, l_{j}\}$ by $j$ for any $0 \le j \le 3n-1$. The labels are not identifier but just position indicators.
Then we can view $A_j$ as a graph with all the edges but edges with labels $j$, $n+j$, $2n+j$ missing in both $E_L$ and $E_R$.
For simplicity, let us call any missing edge in $E_L$ or $E_R$ a ``gap''.
Note that the number of gaps a 3-partition $X$ has will be 4, 5, or 6, if the number of line component in $X$ is 2, 1, or 0, respectively.
Define the \emph{average gap position} of a 3-partition to be the sum of all gap positions divided by the number of gaps, 
and denote by $\mathcal{B}_x$ the set of 3-partitions whose average gap position is $x$.
Note that $A_j \in \mathcal{B}_{n+j}$ for $0 \le j \le n-1$. 
Let $\mathcal{C} = \{C=(C_1, C_2, C_3) \mid E(C_i) \cap E_= = \emptyset \text{ for some } i\}$.  

Let $\Omega_{\text{LEFT}} = \mathcal{B}_n \setminus \mathcal{C}$, $\Omega_{\text{MID}} = \mathcal{C}$, and $\Omega_{\text{RIGHT}}$ be the rest of the $\Omega$. In particular, for $1 \le j \le n-1$, $A_j \in \mathcal{B}_{n+j} \setminus \mathcal{C} \subseteq \Omega_{\text{RIGHT}}$.


First, we argue that removing $\Omega_{\text{MID}}$ from the state space $\Omega$ results in a disconnected graph.
The idea is to argue that, starting with any state $\Omega_{\text{LEFT}} = \mathcal{B}_n \setminus \mathcal{C}$, the chain can only go to states with average gap position exactly equal to $n$, without entering any states in $\Omega_{\text{MID}} = \mathcal{C}$.

Suppose we are taking a single step in ReCom from $X = (X_1, X_2, X_3)$ to $Y = (Y_1, Y_2, Y_3)$, both without line components, i.e. $X, Y \not\in \mathcal{C}$. W.l.o.g we can assume that after one step of ReCom, $X_1$ and $X_2$ are merge-split into $Y_1$ and $Y_2$, i.e. $X_3 = Y_3$. Then the only possible moves have the following pattern: take $k$ (some integer) edges in $E_L$ from $X_1$ to $X_2$, while take $k$ edges in $E_R$ from $X_2$ to $X_1$, with the restriction that the $E_L$-side gap and $E_R$-side gap between $X_1$ and $X_2$ move in the opposite directions but with the same amount. Therefore, the average gap position won't be changed.
Note that it is not hard to see that going into any states in $\mathcal{C}$ would easily shift the average gap position.

Second, we show that the conductance $\Phi$ is inversely exponentially small.
Given that $\Omega_{\text{LEFT}} = \mathcal{B}_n \setminus \mathcal{C}$ is only connected to the states in $\Omega_{\text{MID}}=\mathcal{C}$, we have
\[Q(\Omega_{\text{LEFT}}, \Omega \setminus \Omega_{\text{LEFT}}) = 
\sum_{\substack{x \in \Omega_{\text{LEFT}}\\ y \in \Omega\setminus \Omega_{\text{LEFT}}}} \mu(x) P(x,y)
= \sum_{\substack{x \in \Omega_{\text{LEFT}} \\ y \in \Omega_{\text{MID}}}} \mu(x) P(x,y)
= \sum_{\substack{x \in \Omega_{\text{LEFT}} \\ y \in \Omega_{\text{MID}}}} \mu(y) P(y,x)
\le 
\mu(\Omega_{\text{MID}}).\]
Since $\mu(\mathcal{B}_n) = \mu(\mathcal{B}_{n+j})$ and $\mathcal{B}_{n+j} \in \Omega_{\text{RIGHT}}$ for any $1 \le j \le n-1$, we know that $\mu(\Omega_{\text{LEFT}}) = \mu(\mathcal{B}_n \setminus \mathcal{C}) \le \frac{1}{2}$, and consequently
\[
\Phi = \min_{S: \mu(S) \leq 1/2} \Phi(S) \le \Phi(\Omega_{\text{LEFT}}) = \frac{Q(\Omega_{\text{LEFT}}, \Omega \setminus \Omega_{\text{LEFT}})}{\mu(\Omega_{\text{LEFT}})} \le \frac{\mu(\Omega_{\text{MID}})}{\mu(\Omega_{\text{LEFT}})} \le \frac{\mu(\mathcal{C})}{\mu(A_0)}.
\]
Thus it suffices to show that
$\frac{\mu(\mathcal{C})}{\mu(A_0)}$ is inversely exponentially small in $n$.

Denote by $T(G)$ the number of spanning trees/forests in graph $G$.
From \cite{Daoud14}, we know that for the $2 \times n$ grid graph $G_{2,n}$, $c(a^n-1)\leq T(G_{2,n}) \leq c a^n$ for some constants $a = 2 + \sqrt{3} > 3$ and $c > 0$.
It is not hard to see that for any state $X \in \mathcal{C}$ such that $\phi(X) \le n$, $T(X) \le  p(n) T(G_{2, n}) = O(p(n)a^n)$ for some polynomial $p(\cdot)$, whereas $T(A_0) = T(G_{2,n})^3 = \Omega(a^{3n})$.
Moreover, the number of states in $\mathcal{C}$ is also polynomially upper-bounded in $n$. One can see this by polynomially upper-bounding the number of possible ``gap configurations'', and note that each state has a unique gap configuration.
Therefore, $\frac{\mu(\mathcal{C})}{\mu(A_0)} \le |\mathcal{C}| \min_{X \in \mathcal{C}}\frac{\mu(X)}{\mu(A_0)} = O(\exp(-\theta(n)))$ for some polynomial $\theta(\cdot)$.
\end{proof}

\begin{theorem} \label{thm:grid with hole slow}
ReCom for 3-partitions mixes torpidly on the grid-with-a-hole graphs.
\end{theorem}
\begin{proof}
The grid-with-a-hole graphs are similar to the double-cycle graphs, except at the four corners.
To make sure the labels of ``non-corner'' edges of the ``inner rectangle'' could align with those of the ``outer rectangle'', we label the corner edges unevenly.
Specifically, in contrast to the double-cycle graphs where the labels of two neighboring edges always differ by 1 (mod cycle length), here we let the labels of two neighboring corner edges differ by 2 in the inner rectangle, and by 0 in the outer rectangle. See Figure~\ref{fig:grid_label} for an example.

\begin{figure}[h!]
\centering
\begin{subfigure}[b]{0.315\linewidth}
\centering\includegraphics[width=\linewidth]{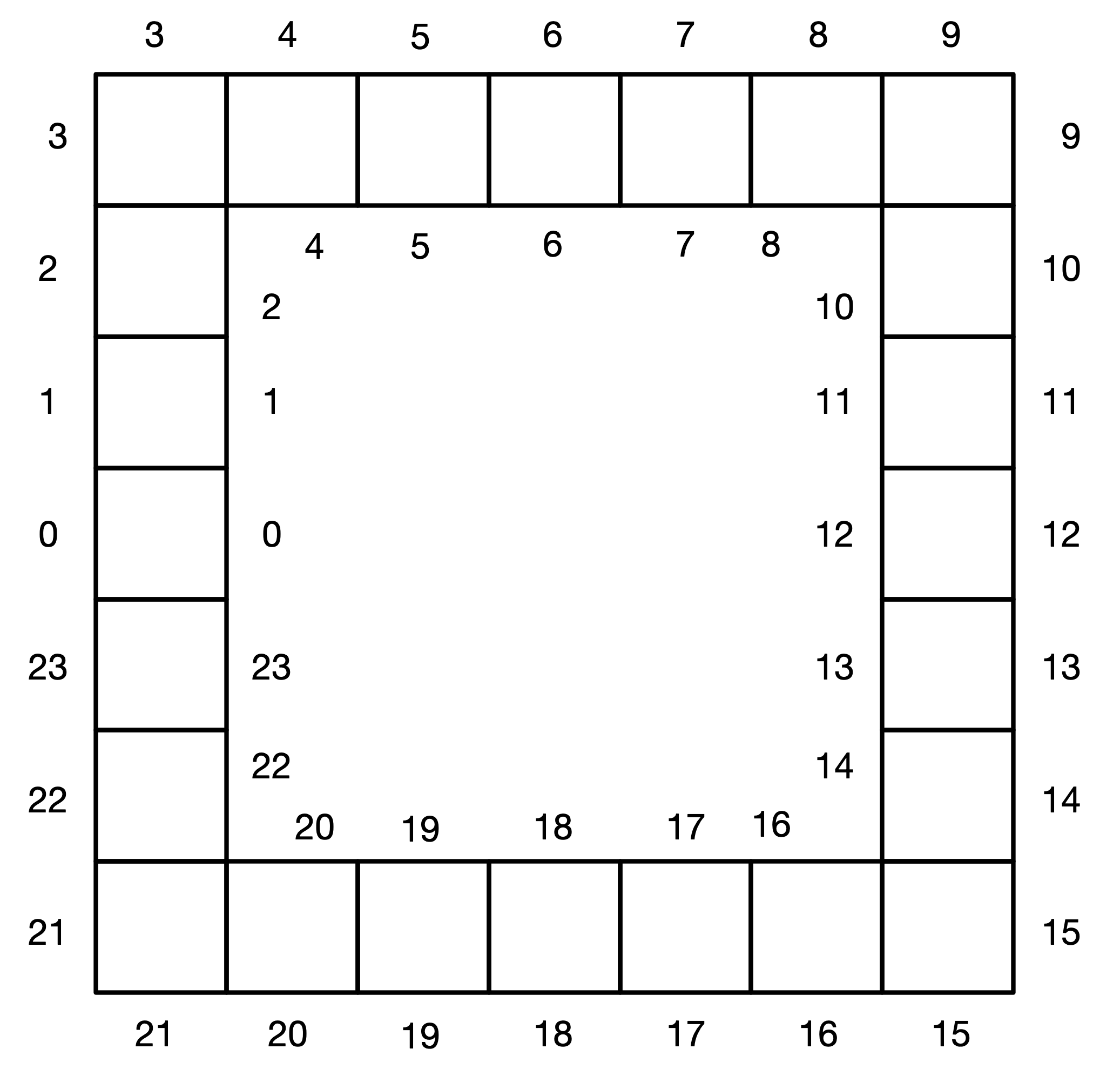}\caption{$G^\square_{8, 8}$}\label{fig:grid_label}
\end{subfigure}
\hspace{0.01\linewidth}
\begin{subfigure}[b]{0.315\linewidth}
\centering\includegraphics[width=\linewidth]{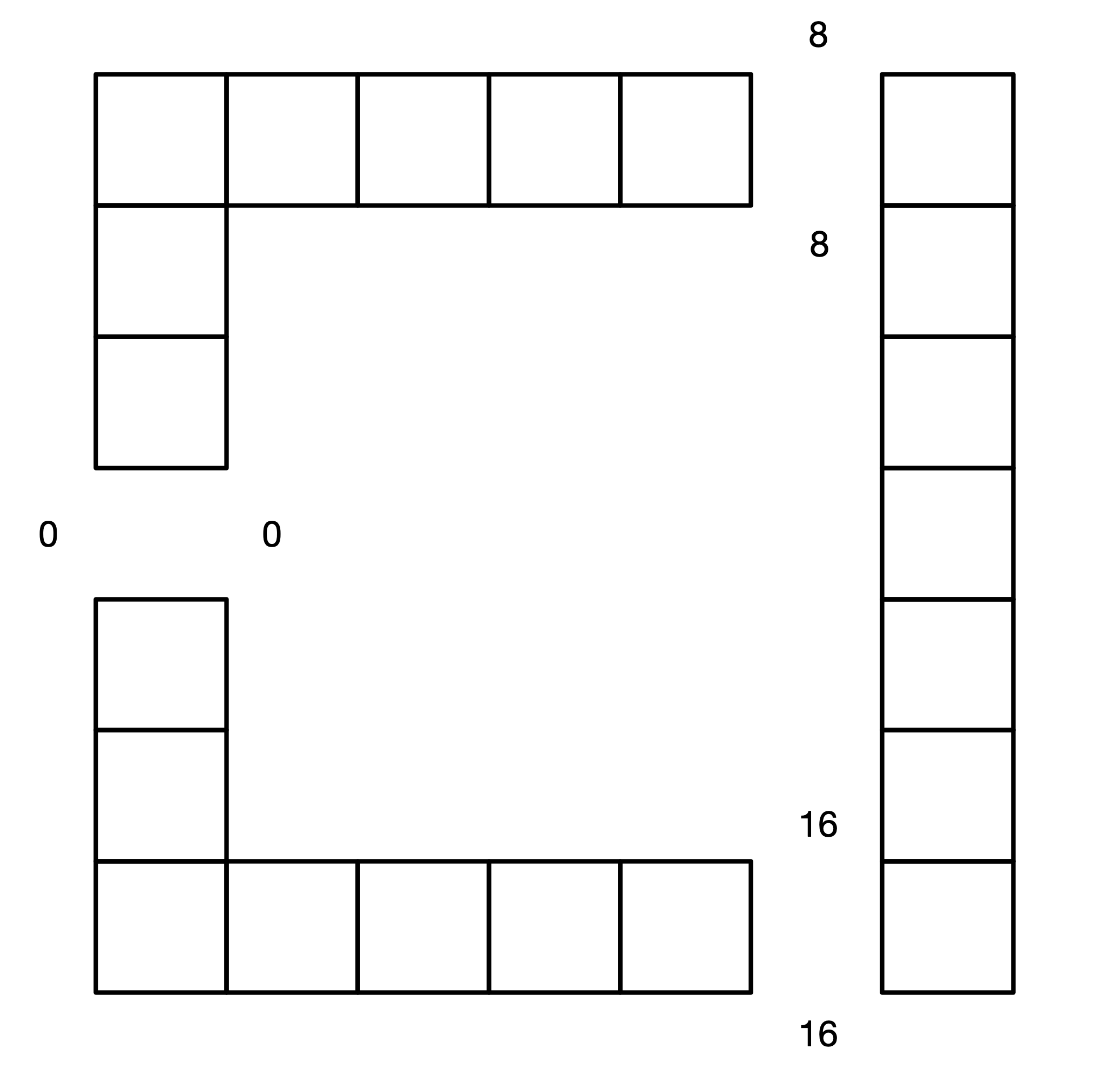}\caption{$A$}\label{fig:grid_a1}
\end{subfigure}
\hspace{0.01\linewidth}
\begin{subfigure}[b]{0.315\linewidth}
\centering\includegraphics[width=\linewidth]{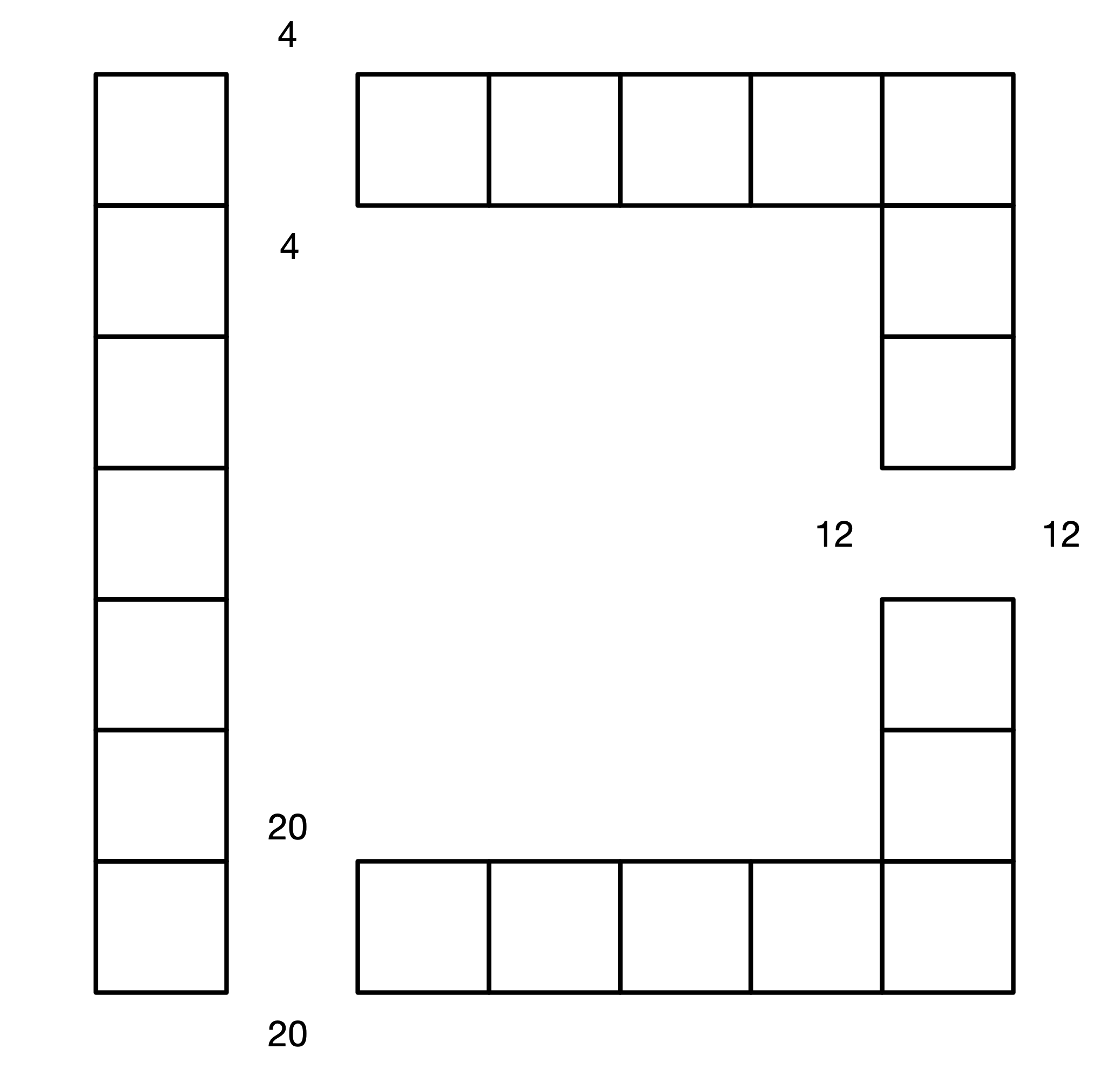}\caption{$A'$}\label{fig:grid_a2}
\end{subfigure}
\caption{The $8 \times 8$ grid-with-a-hole graph $G^\square_{8, 8}$ and two 3-partitions of it.}
\end{figure}

Pick an arbitrary 3-partition $A$ of the maximum weight (see Figure~\ref{fig:grid_a1}).
Let $x_A$ be the average gap position of $A$.
Let $\Omega_{\text{MID}} = \mathcal{C}$ be the 3-partitions in which at least one component lies completely in the inner or outer rectangle.
Let $\Omega_{\text{LEFT}}$ be the connected component containing $A$ in $\Omega \setminus \Omega_{\text{MID}}$ via ReCom transitions,
and let $\Omega_{\text{RIGHT}}$ be the rest of the state space.

First, we show that $\Omega_{\text{RIGHT}}$ is not empty, and in fact $\mu(\Omega_{\text{RIGHT}}) \ge \mu(\Omega_{\text{LEFT}})$.
The idea is to show that by starting from $A$, without entering $\mathcal{C}$, although ReCom can go to some state with average gap position $x$ not necessarily the same as in $A$, $x - x_A$ must be $O(1)$. 
The 6 gaps in $A$ can be partitioned into 3 pairs, with each pair containing one inner gap and one outer gap of the same label.
Without entering $\mathcal{C}$, at each step only one pair will be touched, and the two gaps in the pair move in the opposite directions.
The average gap position could be changed only if at least one of the gap goes through the corner edges, creating a net difference of absolute value at most $2 / 6 = \frac{1}{3}$.
However, due to geometric obstructions, neither of the inner and outer gaps could travel the entire inner/outer rectangle once. Since the number of corners they go through are bounded by a constant, the shift in average gap position is also upper bounded by a constant.

On the other hand, there exists a state $A' \in \Omega_{\text{RIGHT}}$ symmetric to $A$ (see Figure~\ref{fig:grid_a2}), but with $x_{A'} - x_A = \Omega(n)$ where $n$ is the graph size.
Consequently, this means the connected component containing $A'$ in $\Omega \setminus \Omega_{\text{MID}}$ must be in $\Omega_{\text{RIGHT}}$, whose weight is at least as large as $\mu(\Omega_{\text{LEFT}})$.

Then, following a similar argument as in the proof of Theorem \ref{thm:double-cycle}, one can show the conductance of ReCom $\Phi \le \frac{\mu(\mathcal{C})}{\mu(A)}$ is inversely exponentially small in the graph size.
\end{proof}

For completeness, we include a proof of ReCom being irreducible on double-cycle graphs. The fact that its being irreducible on grid-with-a-hole graphs can be proved similarly.
\begin{theorem}
ReCom for 3-partitions is irreducible on double-cycle graphs.
\end{theorem}
\begin{proof}
To prove that the transition graph of ReCom is connected, we show that all the states can move to a subset of states via ReCom transistions, and the subset is inter-connected.

Let $\mathcal{L}$ be the set of 3-partitions with two ``line segments'', one in $E_L$ and the other in $E_R$, whose endpoints align perfectly with each other's. For example, see Figure~\ref{fig:cycle_1} and Figure~\ref{fig:cycle_3}. $\mathcal{L}$ is inter-connected via ReCom moves illustrated in Figure~\ref{fig:cycle}.

\begin{figure}[h!]
\centering
\begin{subfigure}[b]{0.31\linewidth}
\centering\includegraphics[width=0.8\linewidth]{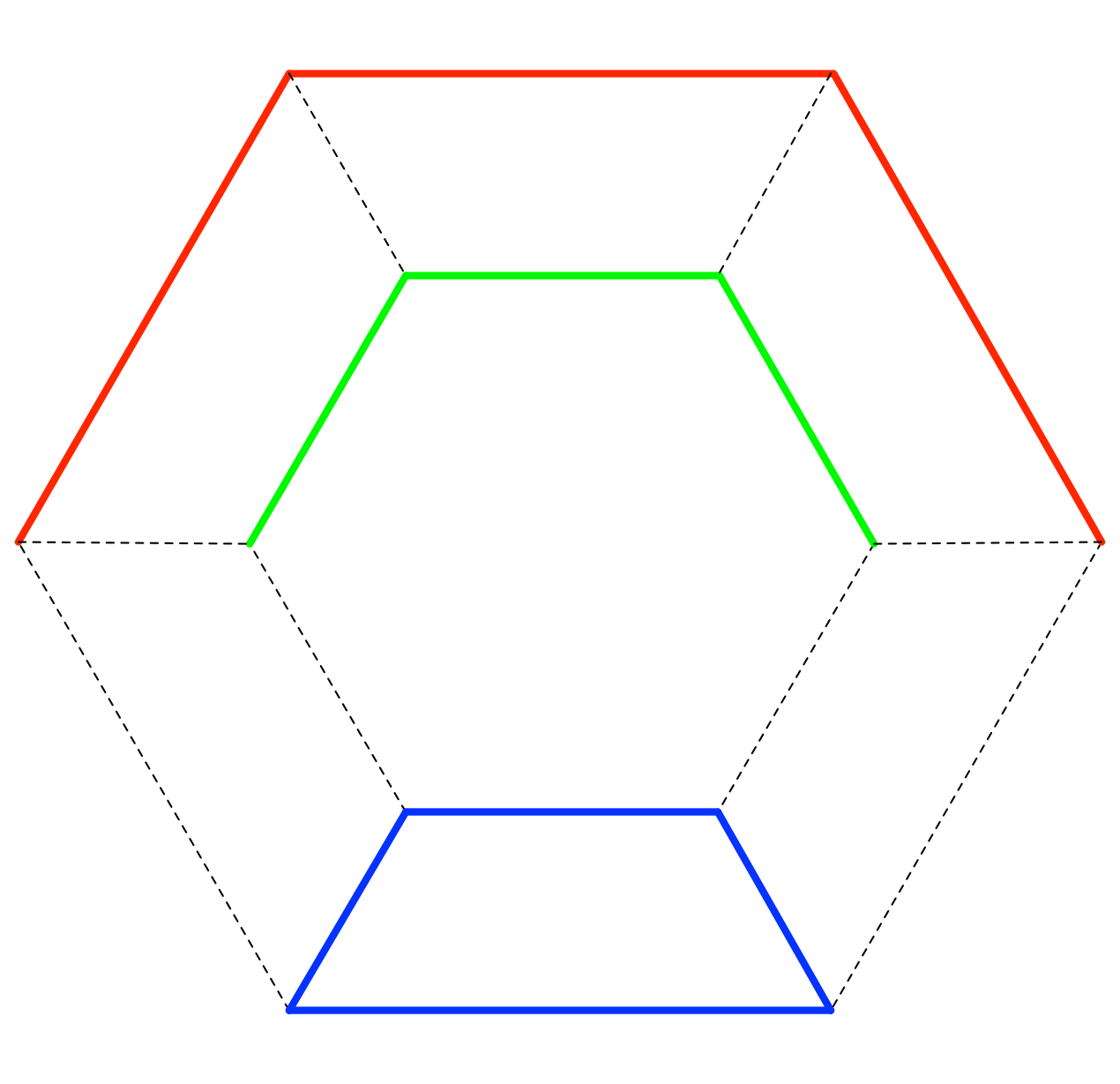}\caption{}\label{fig:cycle_1}
\end{subfigure}
\hspace{0.01\linewidth}
\begin{subfigure}[b]{0.31\linewidth}
\centering\includegraphics[width=0.8\linewidth]{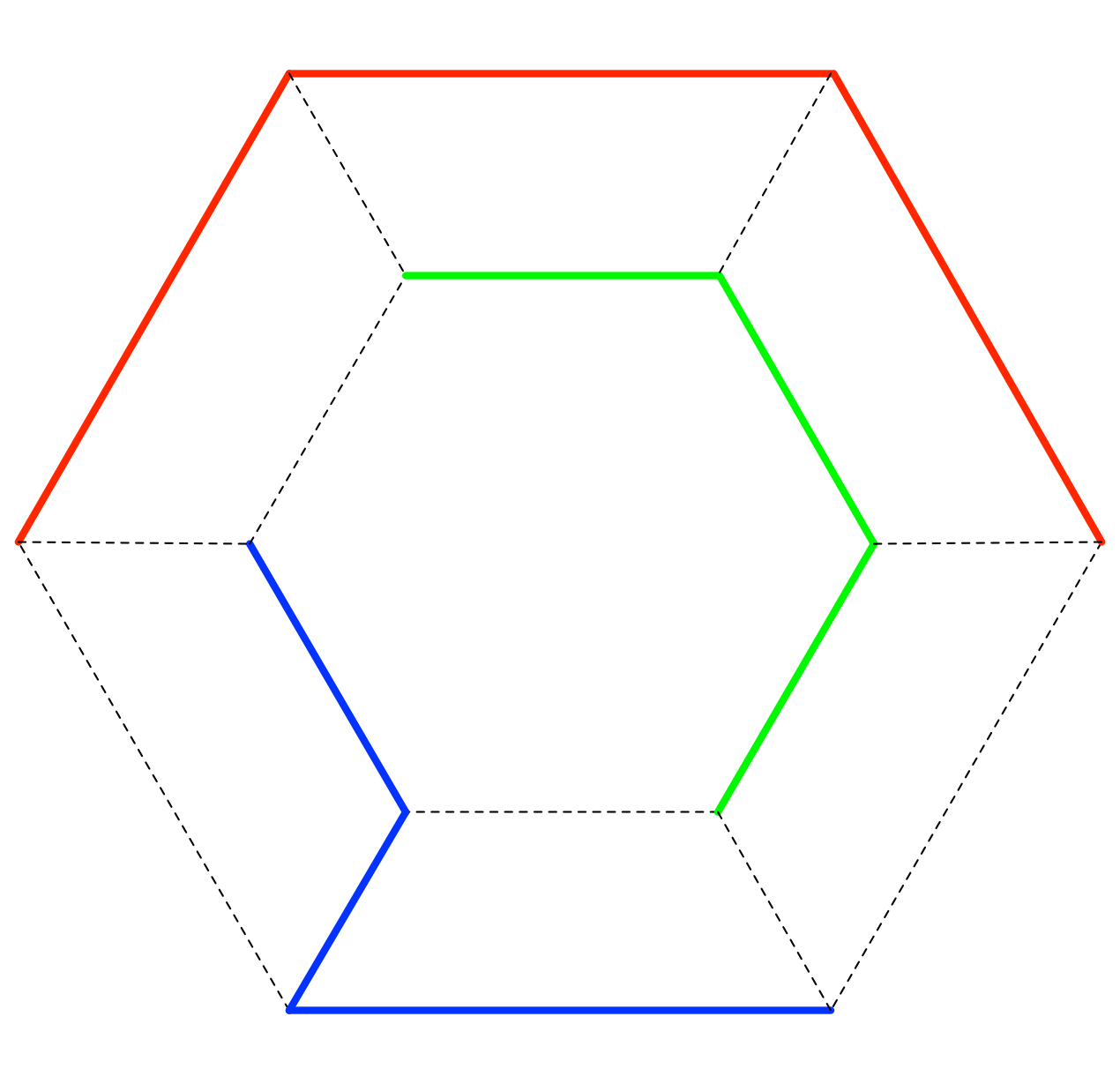}\caption{}\label{fig:cycle_2}
\end{subfigure}
\hspace{0.01\linewidth}
\begin{subfigure}[b]{0.31\linewidth}
\centering\includegraphics[width=0.8\linewidth]{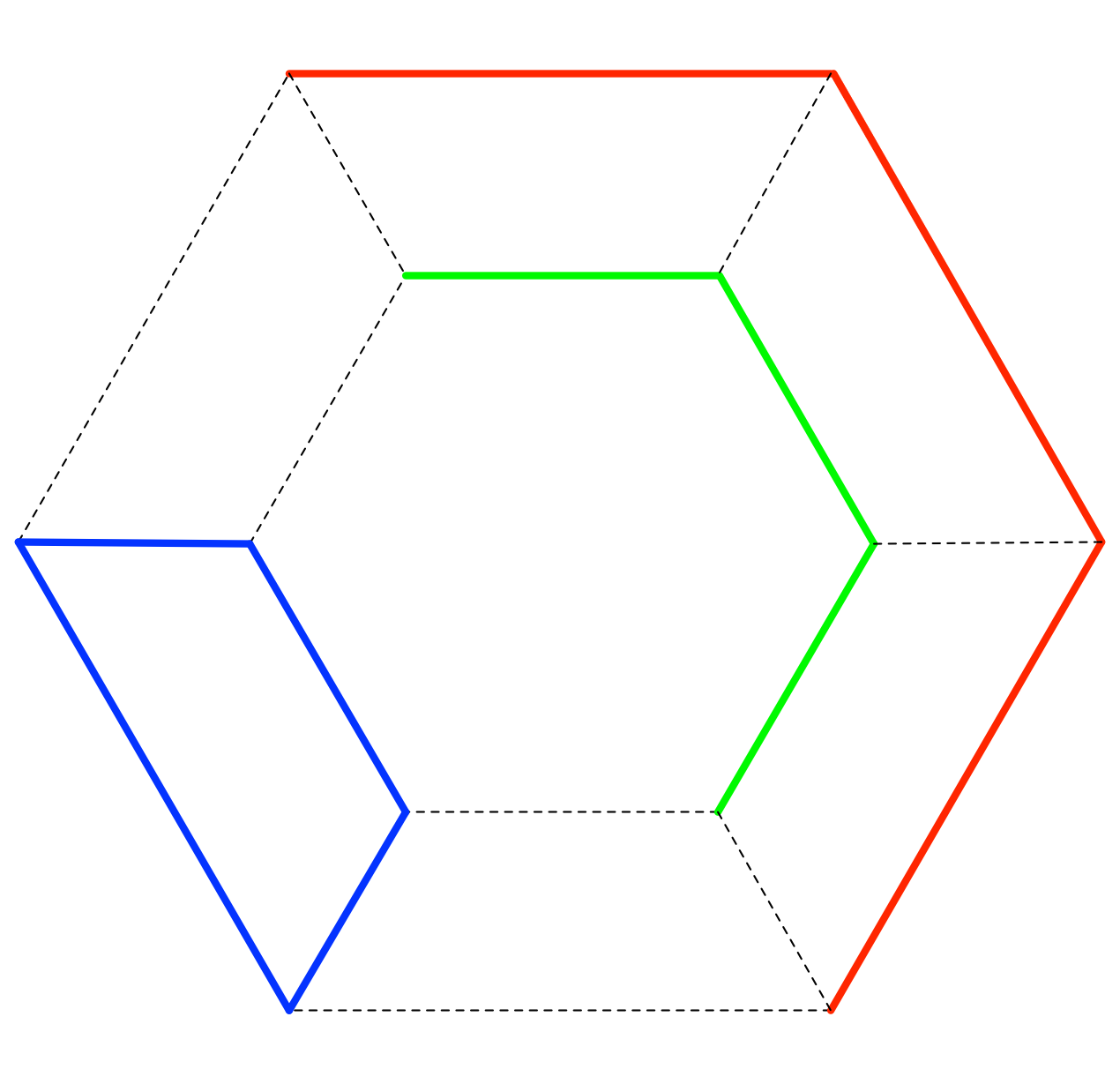}\caption{}\label{fig:cycle_3}
\end{subfigure}
\caption{ReCom transitions (from left to right) which inter-connect $\mathcal{L}$.}\label{fig:cycle}
\end{figure}

For any state $X = (X_1, X_2, X_3) \in \mathcal{C}$, at least one of the $X_i$'s is a line segment. W.l.o.g. assume it is $X_1$. In one step of ReCom, we can merge-split $X_2$ and $X_3$ to arrive at $X' = (X_1, X'_2, X'_3)$ so that $X' \in \mathcal{L}$.
For any state $X = (X_1, X_2, X_3) \not\in \mathcal{C}$, we know that there are 6 gaps which come in pairs. For each pair of gaps, we can equate them using at most one ReCom transition. Thus in three steps, we will reach a state $X' = (X'_1, X'_2, X'_3)$ with $\phi(X') = 3n$. By merge-spliting two components (say $X'_2$ and $X'_3$) in another step of ReCom, we can move into a state in $\mathcal{L}$.
\end{proof}


\section{Fast sampling algorithm from the spanning tree distribution $\mu^*_{G,k}$} \label{sec:relax ReCom}
In this section, we show a variant of ReCom that, in quasi-linear time, samples from the spanning tree distribution without the balanced constraint. Each step of this chain merges two trees and then splits any tree in the ensemble of $(k-1)$ trees (see \cref{def:variant ReCom}).
\begin{definition} \label{def:variant ReCom}
Suppose the current state is a tuple of $k$ trees $T_1, \dots, T_k.$ Each step of the variant ReCom operates by: 
\begin{enumerate}
    \item Add an edge $e$ that connects two different parts $P_i$ and $P_j$ of the partition. This edge, together with the spanning trees $T_i$ and $T_j$ on $G[P_i]$ and $G[P_j]$ respectively, form a spanning tree $T$ for $G[P_i \cup P_j].$
    \item Choose a tree $T'$ among the $(k-1)$ trees $\set{T_r \given r\neq i, j} \cup \set{T}.$ Remove an edge $f$ from $T'.$ 
    This creates two new trees $T'_i$ and $T'_j.$
    Let the new partition be the set of the vertices of the $k$ trees $ \set{T_r \given r\neq i, j} \cup \set{T, T'_i,T'_j} \setminus \set{T'}. $
\end{enumerate}    
\end{definition}

With suitable probability of choosing tree $T'$ and edge $ f$, this variant of ReCom is simply the up-down walk  \cite[see][for details]{ALOVV21} on the distribution of $(n-k)$-edges forest. This walk mixes in $O(m\log m)$-step where $m$ is the number of edges in the graph $G$, and each step can be implemented in $O(\log m)$ time \cite[Theorem 3]{ALOVV21}.

In general, $\mu^c_{G,k}$ is a projection of the distribution $\chi^c_{G,k}$ over forests of $(n-k)$-edges of $G$ defined by
\[\chi^c_{G,k}(F) = \prod_{i=1}^k |P_i|^k\]
where $P_1, \cdots, P_k$ are the connected component of $F.$ Clearly,
\[\mu^c_{G,k} (P_1, \cdots, P_k) = \sum_{F} \chi^c_{G,k}(F)\]
where the sum is over all forests $F$ whose connected components are $P_1, \cdots, P_k.$ We can sample from $\chi^c$ using the up-down walk.
\begin{theorem}
Consider connected graph $G= G(V,E)$ with $ \abs{V} = n$ and $\abs{E} = m.$
The up-down walk for $\mu^c$ mixes in $O(m \log m)$ steps for $ c\in [0,1]$ and in $O(m (n/k)^{3k(c-1)} (n+c)\log m)$-steps for $c>1.$

Each step can be implemented in amortized $O(\log m)$ time for $ c = 0$ and $ O(n) $ time for $c > 0.$ 
\end{theorem}
\begin{proof}
Let $\bar{\chi}^c$ be the complement of $\chi^c$ i.e.
$\bar{\chi}^c (S) = \chi(V \setminus S).$ The up-down walk on $\chi^c$ is precisely the down-up walk \cite[see][]{ALOVV21} on $\bar{\chi}^c.$ For $c\in[0,1],$ the distribution $\bar{\chi}^c$ is log-concave by \cref{lem:chi1} and \cite[Theorem 1.7]{logconcave2}, thus the down-up walk on $\bar{\chi}^c$ has  mLSI constant and Poincare constant $\geq \Omega(m^{-1}).$ 
For $c> 1,$ we can compare the Poincare constant of the down-up walk on $\bar{\chi}^c$ with the down-up walk on $\bar{\chi}^1.$ Let  $\mathcal{P}^{(c)}, \pi^c$ be the transition matrix and stationary distribution for the down-up walk on $\bar{\chi}^c.$ Note that the union of edges in a tuple of $ k$ disjoint trees forms a forest containing $n-k$ edges.
Each step operates on a forest $F$ by adding a uniformly random edge $e\not\in F$, then removing an edge $f$ from $F \cup \set{e}$ to create a new forest $F' = T'_1 \cup \cdots T'_k$  with probability proportional to  $\prod_{i=1}^k \abs{T_i}^c.$ 

 Note that
\[Z_{\chi^c_k} = \sum_{(T_1, \cdots, T_k)} \prod_{r=1}^k \abs{T_i}^c \leq (n/k)^{(c-1)k} \sum_{(T_1, \cdots, T_k)} \prod_{r=1}^k \abs{T_i} = (n/k)^{(c-1)k} Z_{\chi^1_k}\]
Consider two tuples of spanning tree $ \mathcal{T}= (T_1, \cdots, T_k)$ and $\mathcal{T}'= (T'_1, \cdots, T'_k).$ Let $ F= \bigcup_i T_i$ and $F'=\bigcup_i T'_i.$ For $  \mathcal{P}^{(c)} (\mathcal{T}, \mathcal{T}') > 0,$ we need $F'= F\cup \set{e} \setminus \set{f}$ for some edges $e,f.$ Let $\mathcal{R}$ be the set of $(n-k)$-edges forests that can be obtained from removing an edge from $F\cup F' = F\cup \set{e}.$   
\begin{align*}
    \pi^c(T) \mathcal{P}^{(c)} (\mathcal{T}, \mathcal{T}')  &= \frac{\chi_{G,k}^c (F) }{Z_{\chi^c_k} } \cdot \frac{\chi_{G,k}^c (F')  }{(m-n+k)\sum_{W} \chi_{G,k}^c (W)} \\
    &\geq \frac{\chi_{G,k}^1 (F) \chi_{G,k}^1(F')}{(n/k)^{2(c-1)k} (m-n+k) Z_{\chi^1_k} \sum_{W} \chi_{G,k}^c (W)}\\
    &= \frac{1}{(n/k)^{2(c-1)k} }  \pi^1(T) \mathcal{P}^{(1)} (\mathcal{T}, \mathcal{T}') 
\end{align*}
where the inequality follows from $ (n/k)^{(c-1)k} \chi_{G,k}^1 (W)\geq \chi_{G,k}^c (W) \geq \chi_{G,k}^1 (W) . $
Next
\[\pi^c(T) = \frac{(\prod_{r} \abs{T_i})^c } {Z_{\chi^c_k} } \leq \frac{(n/k)^{(c-1)k}(\prod_{r} \abs{T_i}) } {Z_{\chi^1_k} }  =(n/k)^{(c-1)k}  \pi^1 (T)\]
Thus, for any function $f$ over the state space of these two chains
\[\mathcal{E}_{\mathcal{P}^{(c)}} (f, f) \geq (n/k)^{2(c-1)k}\mathcal{E}_{\mathcal{P}^{(1)}} (f,  f)   \]
thus by \cite[Equation (2.3)]{DS94},
\[\alpha_0(\mathcal{P}^{(c)}) \geq (n/k)^{3k(c-1)}\alpha_0(\mathcal{P}^{(1)}). \]
Moreover
\[\min_{\mathcal{T}} \pi^c(\mathcal{T}) \geq \frac{1}{Z_{\chi^c_k}} \geq ((n/k)^{ck} m^{n-k})^{-1} \]
Applying \cref{thm:mixing} gives the desired bound on the mixing time of $\mathcal{P}^{(c)}.$

For $ c= 0,$ each step of the up-down walk can be implemented in amortized $O(\log m)$ time using link-cut tree \cite[see][]{ALOVV21}. For other value of $c,$ to implement each step of the Markov chain, we can iterate over all edges contained in the tree $T_i \cup T_j \cup \set*{e}$ and computing the size of the two new trees created by removing that edge in $O(n)$ time by running DFS. 

\end{proof}
\begin{lemma} \label{lem:chi1}
The distribution $\chi^1_{G,k}$ is strongly Rayleigh.
\end{lemma}
\begin{proof}
Consider the graph $G'$ obtained from $G$ by adding a node $s$ and edges $(s,v)$ for all $v\in G.$ Consider the uniform distribution over spanning trees of $G'$ that contains exactly $k$ edges incidence to $s.$ This distribution induce the distribution $\chi^1_{G,k}$ over the $(n-k)$-edges forest of $G$, via the map that takes spanning tree $T'$ of $G'$ and maps it to $T' \cap E(G),$ which forms a forest of $G.$ Indeed, for each forest $F$ of $G$ with $(n-k)$ edges with connected components $P_1,\cdots, P_k,$ there is exactly $\abs{P_1} \times \cdots \abs{P_k}$ ways to adding $k$ edges from $s,$ with $1$ edge from each component $P_i$, to form a spanning tree of $G'.$ 

The uniform distribution $\nu$ over spanning tree of $G'$ is strongly Rayleigh. The generating polynomial $f_{\nu}$ of $\nu$ is
\[f_{\nu} (\set*{z_e}_{e\in E(G)}, \set*{z_{(s,v)}}_{v\in G}) = \sum_{T \text{ spanning } G'} \prod_{e' \in T} z_{e'} \]
Setting $z_{(s,v)}$ to $z$, we still have a homogeneous real stable polynomial. Taking the part with degree $k$ in $z$ gives the generating polynomial for $\chi^1_{G,k},$ which is strongly Rayleigh by \cite[Theorem 3.4]{Choe_2004}.  
\end{proof}
\section{Fraction of balanced partitions}
\label{sec:FractionBalanced}
In this section, we provide evidence for \cref{conj:percentage of balance partition}, and our goal is to lower bound $\frac{Z_{\mu^{\text{balanced}}_k} }{Z_{\mu^{*}_k} }.$

First, we bound the partition function of $\mu^*_k$ in terms of the number of spanning trees.
\begin{lemma}\label{lem:partition function bound}
For a connected graph $G = G(V,E)$
\[Z_{\mu^*_k} \leq \binom{n-1}{k-1} T(G) \leq n^{k-1} T(G).\]
\end{lemma}
\begin{proof}
Recall that $Z_{\mu^*_k}$ is exactly the number of forests containing exactly $(n-k)$ edges in $G.$

Pick an arbitrary map $\Phi$ that map $(n-k)$-edges forest $F$ to a spanning tree $T$ of $G$ that contains $F$ (such $T$ always exists because $G$ is connected). For each spanning tree $T$ of $G,$ there are $\binom{n-1}{k-1}$ ways to choose a preimage of the map $\Phi$ i.e. removing $k-1$ edges from $T$ to obtain a forest $(n-k)$-edges $F$ contained in $T.$ The desired inequality follows.
\end{proof}

For graph $G$ and $H,$ let its Cartesian product $G \times H$ be the graph with vertices $(u, v)\in V(G) \times V(H)$ and edges
\[E(G\times H) = \set*{((u,v), (u',v)) \given (u,u') \in E(G)}  \cup \set*{((u,v), (u,v')) \given (v,v') \in E(H)}.\]
For example, the rectangular grid graph $G_{m,n}$ is $P_m \times P_n$ with $P_n$ be the path graph with $n$ vertices.  
We also consider the family of cylindrical grid graph $C_{m,n} = P_m \times C_n,$ which includes the double-cycle graph $C_{2,n}.$ 
\begin{theorem} \label{thm:percentage of balance partition simple graph}
For the rectangular graph $G_{m, n}$ and the cylindrical grid graph $C_{m,n} =P_m \times C_n$ with $m = O(1)$, the fraction of balanced $k$-partitions $Z_{\mu^{\text{balanced}}_k} / Z_{\mu^{*}_k}$ is lower bounded by $\Omega\left(\frac{1}{n^{k-1} \Theta(1)^k }\right)$ and $\Omega\left(\frac{1}{n^{k} \Theta(1)^k}\right)$, respectively.
\end{theorem}
\begin{proof}
W.l.o.g., we can assume $k \mid n$.
\cite[Theorem 2 and 4]{Daoud14} shows that for $m = O(1),$ 
\[T(G_{m,n}) = \Theta (c^n) \]
and
\[T(C_{m,n}) =n \Theta(c^n)\]
for some constant $c > 1$ depending on $m.$
On the other hand, the $k$-partition $X$ of $G_{m,n}$ or $C_{m,n}$, with each component being $G_{m, n/k},$ has weight
\[\mu^{\text{balanced}}_k(X) = \prod_{i=1}^k\Theta(c^{n/k}) = \Theta(c^n) \Theta(1)^k,\]
and thus 
\[Z_{\mu^{\text{balanced}}_k}\geq \mu^{\text{balanced}}_k(X) = \Omega(c^n).\]
This, together with \cref{lem:partition function bound}, implies the desired inequality.
\end{proof}

We remark that similar bounds hold for torus graphs, which we omit as the focus of this paper is on planar graphs.

By slightly modifying the proof of Theorem~\ref{thm:percentage of balance partition simple graph}, we can show similar results for the grid-with-a-hole graphs.
\begin{theorem}
For the $m \times n$ grid-with-a-hole graph $G^\square_{m, n}$, the fraction of balanced $k$-partitions $Z_{\mu^{\text{balanced}}_k} / Z_{\mu^{*}_k}$ is lower bounded by $\Omega\left(\frac{1}{(m+n)^k \Theta(1)^k}\right)$.
\end{theorem}
\begin{proof}
To get an upper bound on $T(G^\square_{m, n})$, we can slightly modify the four corners of $G^\square_{m, n}$ by adding some vertices and edges, as is shown in Figure~\ref{fig:corner_1} to \ref{fig:corner_2}. Note that such process turns $G^\square_{m, n}$ into $C_{2, 2(m+n-2)}$, while not decreasing the number of spanning trees. This can be shown by constructing an injective map from the set of spanning tree in $G^\square_{m, n}$ to that of $C_{2, 2(m+n-2)}$. Therefore, we have
\[T(G^\square_{m, n}) \le T(C_{2, 2(m+n-2)}) = n \Theta(c^{2(m+n)})\] for the same $c > 1$ in the proof of Theorem~\ref{thm:percentage of balance partition simple graph}.

\begin{figure}[h!]
\centering
\begin{subfigure}[b]{0.31\linewidth}
\centering\includegraphics[width=0.8\linewidth]{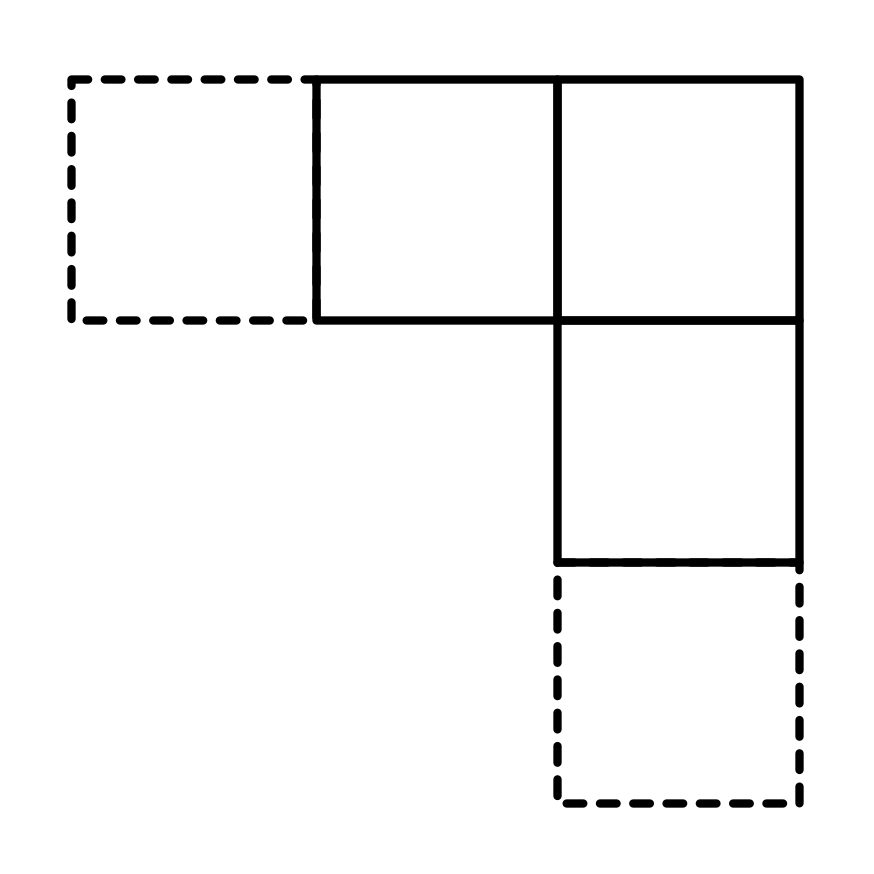}\caption{}\label{fig:corner_1}
\end{subfigure}
\hspace{0.01\linewidth}
\begin{subfigure}[b]{0.31\linewidth}
\centering\includegraphics[width=0.8\linewidth]{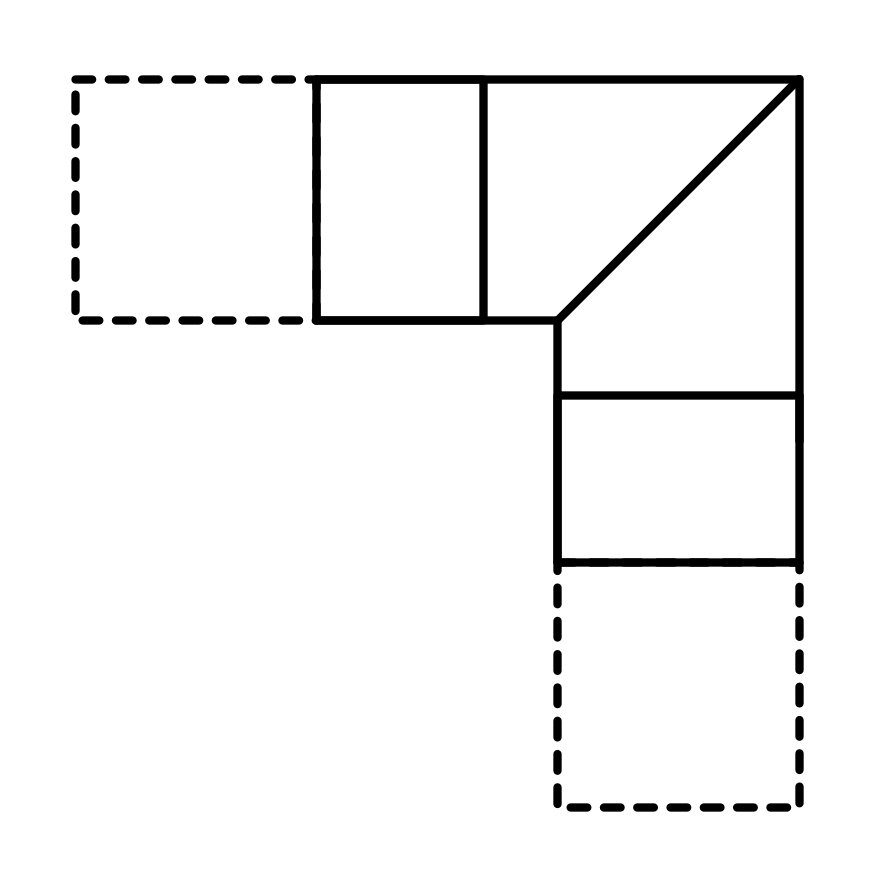}\caption{}\label{fig:corner_2}
\end{subfigure}
\hspace{0.01\linewidth}
\begin{subfigure}[b]{0.31\linewidth}
\centering\includegraphics[width=0.8\linewidth]{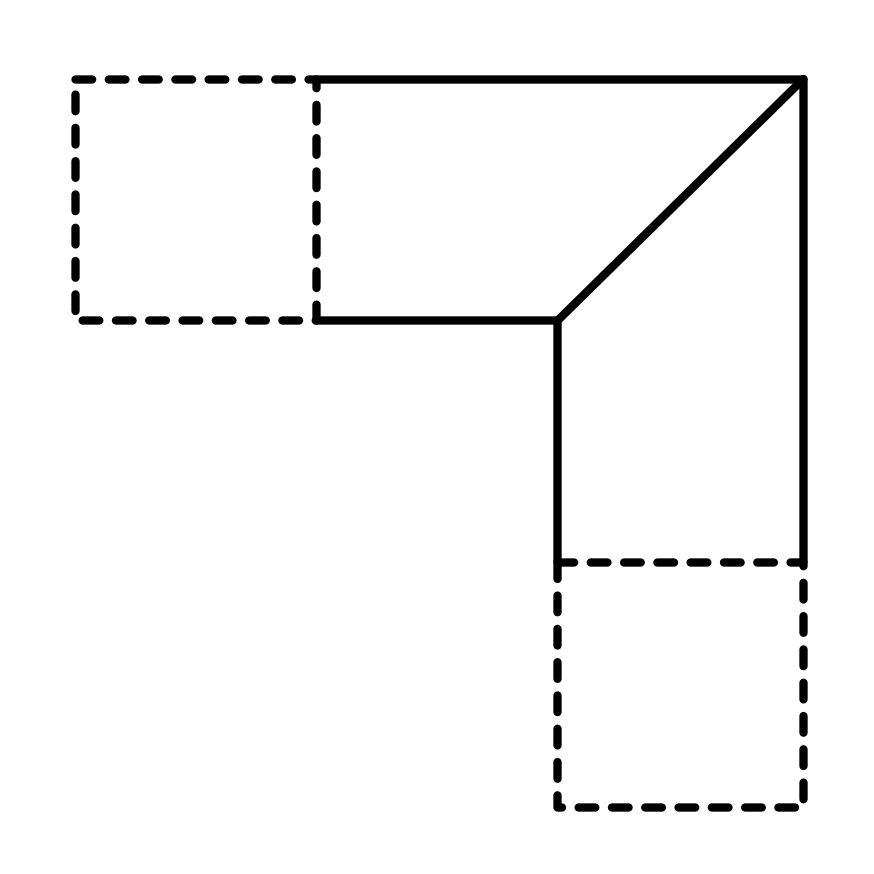}\caption{}\label{fig:corner_3}
\end{subfigure}
\caption{Modifications at the corners of grid-with-a-hole graphs.}\label{fig:corner}
\end{figure}

On the other hand, pick any maximum-weight $k$-partition $X = (X_1, X_2, \dots, X_k)$ of $G^\square_{m, n}$, of which each component $X_i$ should contain $2(m+n) / k - O(1)$ edges between the inner rectangle and the outer rectangle of $G^\square_{m, n}$.
To get an lower bound on any component $X_i$, we again slightly modify the corners in $X_i$ (if any) by merging/contracting some vertices and edges, as is shown in Figure~\ref{fig:corner_1} to \ref{fig:corner_3}. This process does not increase the number of spanning trees, while turning $X_i$ into a subgraph $Y_i$ of the double-cycle graph $C_{2, 2(m+n-4)}$ having $2(m+n) / k - O(1)$ edges between the two cycles. Since $T(Y_i) = \Theta(c^{2(m+n)/k})$ for the same $c$, we know that
\[\mu^{\text{balanced}}_k(X) = \prod_{i=1}^k\Theta\left(c^{2(m+n)/k}\right) = \Theta\left(c^{2(m+n)} \Theta(1)^k\right).\]
The rest of the proof follows from that of Theorem~\ref{thm:percentage of balance partition simple graph}.
\end{proof}
\PrintBibliography
\end{document}